\documentclass[11pt]{article}
\usepackage{fullpage}
\usepackage[utf8]{inputenc}
\usepackage[T1]{fontenc}

\usepackage{graphicx}
\usepackage{amsmath,amsthm,amssymb}
\usepackage{algorithm}
\usepackage{algpseudocode}
\usepackage{multirow}
\usepackage{makecell}

\usepackage{thm-restate,color,xspace}
\usepackage{comment}
\usepackage{thmtools}
\usepackage{xcolor}
\usepackage{nameref}
\usepackage{array}
\usepackage{mathrsfs}
\usepackage{enumerate}
\usepackage{enumitem}

\definecolor{ForestGreen}{rgb}{0.1333,0.5451,0.1333}
\definecolor{DarkRed}{rgb}{0.65,0,0}
\definecolor{Red}{rgb}{1,0,0}
\usepackage[linktocpage=true,
pagebackref=true,colorlinks,
linkcolor=DarkRed,citecolor=ForestGreen,
bookmarks,bookmarksopen,bookmarksnumbered]
{hyperref}
\usepackage[capitalize]{cleveref}

\usepackage{xcolor}
\usepackage{nameref}

\usepackage{cleveref}

\makeatother

\usepackage{babel}
\usepackage{hyperref}

\declaretheorem[numberwithin=section]{theorem}
\declaretheorem[numberlike=theorem]{lemma}

\declaretheorem[numberlike=theorem]{observation}

\declaretheorem[numberlike=theorem,style=definition]{remark}

\declaretheorem[numberlike=theorem,name=Question]{question}

\Crefname{question}{Question}{Questions}

\newcommand{\ignore}[1]{}

\global\long\def\level{{\sf level}}
\global\long\def\dist{{\sf dist}}

\global\long\def\OPT{\mathsf{OPT}}

\global\long\def\calM{\mathcal{M}}

\global\long\def\inh{{\rm inh}}
\global\long\def\orig{{\rm orig}}
\global\long\def\scrC{\mathscr{C}}
\global\long\def\LP{{\rm LP}}
\global\long\def\cost{{\sf cost}}
\global\long\def\Ball{{\sf Ball}}
\global\long\def\ActUn{{\sf ActUn}}

\usepackage[colorinlistoftodos]{todonotes}
\definecolor{debianred}{rgb}{0.84, 0.04, 0.33}

\title{Online Steiner Forest with Recourse}
\author{Yaowei Long\thanks{University of Michigan, \texttt{yaoweil@umich.edu}. The work was done while the author was an intern at Microsoft Research.} 
\and
Sepideh Mahabadi\thanks{Microsoft Research, \texttt{smahabadi@microsoft.com}.}
\and
Sherry Sarkar\thanks{Carnegie Mellon University, \texttt{sherrys@andrew.cmu.edu}.}
\and
Jakub Tarnawski\thanks{Microsoft Research, \texttt{jakub.tarnawski@microsoft.com}.}
}

\date{}
\begin{document}

\maketitle
\thispagestyle{empty}

\hypersetup{pageanchor=false}

\begin{abstract}
In the online Steiner forest problem we are given a graph $G$, and a sequence of terminal pairs $(u_i,v_i)$ which arrive in an online fashion. We are asked to maintain a low-cost subgraph in which each $u_i$ is connected to $v_i$ for all the pairs that have arrived so far.
If we are not allowed to delete edges from our solution, then the best possible competitive ratio is $\Theta(\log n)$. In this work, we initiate the study of low-recourse algorithms for online Steiner forest. We give an algorithm that maintains a constant-competitive solution and has an amortized recourse of $O(\log n)$, i.e., inserts and deletes $O(\log n)$ edges per demand on average.
\end{abstract}

\newpage
\clearpage
\hypersetup{pageanchor=true}
\pagenumbering{arabic}

\section{Introduction}

The \textbf{Steiner tree} problem is a classic question in network design.
Given a weighted graph and a subset of vertices called {\em terminals}, it asks to compute the cheapest tree that connects the terminals.
Since being included on Karp's list of 21 NP-complete problems~\cite{Karp72}
it has been extensively studied.
The problem
is known to be APX-hard~\cite{chlebik2008steiner};
however, getting a constant-factor approximation is straightforward, as computing a minimum spanning tree on the set of terminals already yields a 2-approximation~\cite{KouMB81}.
The current best approximation guarantee is $\ln{4} + \varepsilon <1.39$~\cite{byrka2010improved}.

In the \textbf{online} version of the \textbf{Steiner tree} problem,
introduced in a seminal work of Imase and Waxman~\cite{imase1991dynamic},
the terminals arrive online,
and after each arrival the algorithm must produce a solution that connects the current terminal set
--
without knowledge of the future arrivals
and without the possibility of deleting already added edges.
They
showed that the natural greedy algorithm is $O(\log n)$-competitive,
and that no algorithm can obtain a better competitive ratio.

Motivated by this impossibility result,
Imase and Waxman~\cite{imase1991dynamic} also asked the question of whether deleting a small number of edges
would allow one to maintain a constant-competitive solution.
They showed how to maintain such a Steiner tree
while making $O(n^{3/2})$ changes per $n$ arrivals
(which improves upon the $O(n^2)$ changes needed if one naively recomputes the tree  after each arrival).
In other words, they gave a constant-competitive algorithm with \emph{$O(\sqrt{n})$ amortized recourse}.
This was eventually improved over 20 years later by Megow, Skutella, Verschae, and Wiese~\cite{MegowSVW12}
to constant amortized recourse.
Finally, Gu, Gupta, and Kumar~\cite{GuGK13} showed that it is in fact enough to make just a single edge swap per arrival.

In the more general \textbf{Steiner forest} problem,
instead of a set of terminals we are given a set of {\em terminal pairs},
and asked to compute the cheapest forest that connects the vertices of each terminal pair together.
The first (offline) approximation algorithm for the problem is a primal-dual approach by Agrawal, Klein, and Ravi \cite{agrawal1995trees} that yields a $2$-approximation;
this was further generalized by Goemans and Williamson~\cite{goemans1995general},
and the approximation ratio was improved only very recently to $2-2^{-11}$~\cite{ahmadi2025breaking}
and subsequently to
1.994~\cite{GuptaT25}.
As for algorithms based on more combinational techniques,
a natural greedy algorithm that repeatedly connects the closest yet-unconnected terminal pair
is no better than $\Omega(\log n)$~\cite{ChenRV10},
but Gupta and Kumar~\cite{DBLP:conf/stoc/Gupta015} showed that the \emph{gluttonous algorithm},
which repeatedly connects the two closest terminals (possibly not from the same pair),
does yield a constant-factor approximation.
Furthermore,
Groß, Gupta, Kumar, Matuschke,
Schmidt, Schmidt, and Verschae~\cite{GrossGKMSSV18}
gave a constant-factor approximation based on local search.

The focus of this work is the \textbf{online} version of the \textbf{Steiner forest} problem,
introduced by Westbrook and Yan~\cite{WestbrookY95},
who gave an $\tilde{O}(\sqrt{n})$-competitive algorithm.
Next, Awerbuch, Azar, and Bartal~\cite{AwerbuchAB96}
showed that the greedy algorithm is $O(\log^2 n)$-competitive
using a primal-dual analysis.
Finally, Berman and Coulston~\cite{BermanC97}
gave a primal-dual $O(\log n)$-competitive algorithm.
(This is optimal due to the abovementioned hardness result for online Steiner tree~\cite{imase1991dynamic}.
Their algorithm is not greedy;
the competitive ratio of greedy is a prominent open problem, see e.g.~\cite{bamas2022improved}.)

In contrast to Steiner tree, for Steiner forest
no $o(\log n)$-competitive, low-recourse algorithm has been known.
Finding such an algorithm has been listed as ``an interesting open problem'' in~\cite{gupta2014online}
and ``still wide open'' in~\cite{GrossGKMSSV18}.

\subsection{Our Results}

In this work we give the first low-recourse constant-competitive algorithm for online Steiner forest.

\begin{restatable}{theorem}{Thmmain}
\label{thm:main}
    There is an algorithm for online Steiner forest
    that, over the course of $n$ arrivals
    of terminal pairs $(u_i,v_i)$
    in a metric space,
    maintains an $O(1)$-competitive solution
    while inserting and deleting at most $O(n \log n)$ edges.
\end{restatable}

Our approach in fact yields a tradeoff: for any parameter $\lambda\in[1,\log n]$, it achieves competitive ratio $O(\log(n)/\lambda)$ and total recourse $O(n\lambda)$. See \Cref{thm:FullMain} for a detailed version of \Cref{thm:main}.

\begin{remark}
    As a byproduct of our approach, we also get an $O(\log n)$-competitive algorithm for the online Steiner forest problem (the classic, no-recourse setting).
    At a high level, this follows by neglecting to delete the edges that our algorithm decides to delete,
    and roughly corresponds to simulating the gluttonous algorithm in an online fashion (see \Cref{sect:OnlineGlut} for a formal algorithm description).
    This matches the optimal guarantee of Berman and Coulston~\cite{BermanC97} using a fully combinatorial algorithm
    as opposed to their primal-dual one.
\end{remark}

\subsection{Related Work}

Online algorithms with low recourse
(often called \emph{consistent})
are
a very active and growing research area. Such algorithms have been proposed for many problems:
clustering and facility location \cite{LattanziV17, GuoKLX20, CAHP19, FichtenbergerLNS21,BhattacharyaLP22,LackiHGJR24,ForsterS25,chan2025online,buchbinder2025competitively},
scheduling and load balancing~\cite{PW93, Wes00, AGZ99, AwerbuchAPW01, SandersSS09, SkutellaV10, EL14,GuptaKS14,bernstein_et_al:LIPIcs.ITCS.2017.51,GalvezSV20,KrishnaswamyLS23,FoussoulGK24},
matching~\cite{GKKV95, CDKL09, BLSZ14,GuptaKS14,AngelopoulosDJ18,BernsteinHR19,MegowN20,SolomonS21,BhoreFT24},
matroid intersection~\cite{BuchbinderGHKS24},
set cover \cite{GuptaKKP17,AbboudAGPS19,BhattacharyaHN19,GuptaL20,BhattacharyaHNW21,AssadiS21,SolomonU23,SolomonUZ24,BukovSZ25,bhattacharya2025fullydynamicsetcover,solomon2025dynamicsetcoverworstcase},
graph coloring \cite{SolomonW20}, 
edge coloring~\cite{BhattacharyaGW21},
edge orientation \cite{BrodalF99,SawlaniW20,BeraBCG22},
online learning~\cite{JagharghKLV19},
maximal independent sets \cite{Censor-HillelHK16,AssadiOSS18,BehnezhadDHSS19}, 
broadcast range assignment~\cite{dBSS24},
spanners \cite{BaswanaKS12,BhattacharyaSS22},
discrepancy~\cite{GuptaGKKS22},
chasing positive bodies~\cite{BhattacharyaBLS23},
submodular maximization~\cite{DuettingFLNZ24,DuttingFLNSZ25,BuchbinderNW25},
submodular cover~\cite{GuptaL20},
or
knapsack~\cite{BockenhauerKMRSW26}.

\paragraph{Fully dynamic online Steiner tree.}
Imase and Waxman's $O(\sqrt{n})$-recourse result for Steiner tree~\cite{imase1991dynamic}
in fact also holds
in a more general setting where the input sequence consists of terminal insertions as well as deletions.
For this setting,
Łącki, Oćwieja, Pilipczuk, Sankowski, and Zych~\cite{LackiOPSZ15}
gave a constant-competitive algorithm with $O(\log \Delta)$ recourse,
where $\Delta$ is the ratio of the maximum to
minimum distance in the metric.
Gupta and Kumar~\cite{gupta2014online}
improved this to $O(1)$ recourse.
Gupta and Levin~\cite{GuptaL20}
recovered the result of~\cite{LackiOPSZ15}
in a more general setting of submodular cover.
In all of these results, the recourse bounds are amortized.

\subsection{Technical Overview}
A natural attempt at maintaining an online $O(1)$-approximate Steiner forest is to, for each arrival in the input sequence, independently define the current snapshot of the online solution as the outcome of some offline $O(1)$-approximate algorithm for the current instance. This approach will automatically give an $O(1)$-approximate online solution. To  control the recourse, the low-recourse online Steiner tree algorithm \cite{GuGK13} exploits a well-behaved offline algorithm which constructs an offline solution under the guidance of a \emph{clustering procedure}.

To design a low-recourse online Steiner forest algorithm, we also start with this framework, and fortunately, a generalized clustering procedure already exists for the Steiner forest problem \cite{DBLP:conf/stoc/Gupta015}. However, the following fundamental difference between clustering procedures for Steiner tree and Steiner forest problems prevents us from continuing with the same approach as in \cite{GuGK13}.

\paragraph{The Barrier.} On a high level, a clustering procedure will maintain a clustering $\scrC$ (i.e., a partition) of the terminals and iteratively merge two \emph{active} clusters in $\scrC$ which are close enough in the contracted metric $\calM/\scrC$ (we can think of the input metric $\calM$ as a complete weighted graph on the terminals\footnote{In our model, we assume that the input metric $\calM$ has no Steiner node, i.e., $\calM$ is a metric on terminals. See \Cref{sec:prelim} and particularly \Cref{remark:model} for a further discussion.} and of $\calM/\scrC$ as the graph with each cluster contracted into a single vertex). In the clustering procedure for  Steiner tree, clusters are always active, which means that the procedure is equivalent to merging clusters which are close in the original metric $\calM$, and therefore, adding one edge to the solution suffices to simulate this merge. In contrast,
in the clustering procedure for Steiner forest,
some clusters will become inactive (i.e., they will stay in the clustering without participating in further merge operations). Hence, we need to measure the closeness in the contracted metric $\calM/\scrC$, and thus simulating a merge operation between clusters $C_{1},C_{2}\in\scrC$ requires us to add into the solution a $C_{1}$-$C_{2}$ shortest path in $\calM/\scrC$, which potentially consists of many edges.

We note that the difference between the clustering procedures for these two problems is rooted in the inherent difference between Steiner forest and Steiner tree: a Steiner tree solution must be connected, while a Steiner forest solution may not. This is also the fundamental barrier  to overcome when designing Steiner forest algorithms in other settings (e.g., the offline and classic online settings).

In other words, while applying the approach in \cite{GuGK13} to Steiner forest may control recourse with respect to merges and their revocations (that is, when comparing the clustering procedures for two consecutive arrivals $t-1$ and $t$, some merges done at $t-1$ may be revoked, and new merges may be performed at $t$), it remains unclear how to bound the recourse in terms of edge changes.

\paragraph{Step 1: Pinning Edges.} We proceed completely differently from \cite{GuGK13} and use the idea of \emph{pinning edges}.
On a high level, when performing a merge by buying a shortest path, we \emph{pin} the \emph{cheapest} $(1/\log n)$-fraction of edges on this path, where pinning an edge means that it will not be deleted in the future, even if the corresponding merge is later revoked. To bound the recourse
without harming the approximation, we will guarantee the following.
\begin{itemize}
\item The number of pinned edges is linear in $n$ at the end, which means that the total recourse will be bounded by $O(n\log n)$. This is relatively simple: intuitively, we will enforce the set of pinned edges to be acyclic. 
\item The cost of the pinned edges should always be within a constant factor of the optimum, so that we can still guarantee a constant-factor approximation. To this end, an important step is to bound the total cost of all merges throughout the algorithm by $O(\log n\cdot \OPT)$, including those that are later revoked (see the following Step 2 for a detailed discussion). This will bound the total cost of pinned edges by $O(\OPT)$, since whenever we buy a path, the cheapest $(1/\log n)$-fraction of its edges, which we pin, have cost at most $(1/\log n)$-fraction of the path cost.
\end{itemize}

\paragraph{Step 2: Bounding the Total Merging Cost.} We employ a particular clustering procedure from the offline \emph{timed gluttonous algorithm}~\cite{DBLP:conf/stoc/Gupta015}. It proceeds by \emph{levels}, from $0$ to the top. At each level $i$, it iteratively merges two clusters if their distance is in $[2^{i},2^{i+1})$ (the merging cost is at most their distance). Furthermore, the clustering procedure allows some flexibility in choosing the order of level-$i$ merges; intuitively, we prioritize the level-$i$ merges that already existed in the previous arrival. 

For the above clustering procedure, we can show that for a fixed level $i$, the cost of all merging operations at this level throughout the algorithm is bounded by $O(\OPT)$ using a dual-fitting framework. At a high level, the dual-fitting argument is intuitive: we want to maintain a set of source vertices and obtain a feasible dual solution by growing balls around them. The subtlety in the formal analysis is that choosing (and maintaining) the right source vertices requires exploiting how the offline clustering procedures relate to one another across different arrivals.

Finally, the total merging cost is $O(\log n\cdot\OPT)$, since the merges from levels beyond the top $O(\log n)$ will have negligible costs. We emphasize that this total merging cost is not the cost of our online solution. The online solution consists of (i) the output of an offline $O(1)$-approximate algorithm for the current instance, and (ii) the edges pinned up to that point, which have cost $O(\OPT)$ as long as the total merging cost is $O(\log n\cdot \OPT)$.

\subsection{Outline}

In \Cref{sec:prelim}, we introduce the problem setting and notation. In \Cref{sect:Clustering}, we describe an offline clustering procedure from the timed gluttonous algorithm \cite{DBLP:conf/stoc/Gupta015}, which will be used in our online algorithm. In \Cref{sec:our-alg}, we show our algorithm for maintaining a forest based on the clustering procedure while incurring little recourse. Lastly, we discuss open problems in \Cref{sec:future}.

\section{Preliminaries} \label{sec:prelim}

We use $(\calM,D)$ to denote an (offline) Steiner forest instance, where $D=\{(u_{i},v_{i})\mid 1\leq i\leq n\}$ is a set of $n$ demand pairs, and $\calM = (V,\dist_{\calM})$ is a metric on the \emph{terminals} $V = \{u_{i},v_{i}\mid 1\leq i\leq n\}$.
We call $v_i$ the \emph{mate} of $u_i$ (and vice versa). 
We assume $\dist_{\calM}$ takes values at least $1$, but do not require it to be bounded by a polynomial of $n$, so this assumption is without loss of generality by scaling. For ease of presentation, we assume that each terminal participates in exactly one demand pair. The more general case, where a terminal may belong to multiple demand pairs, can be handled without additional difficulty\footnote{More explicitly, if two terminals $s_i$ and $t_j$ are at the same point $v$ of $\mathcal{M}$, we may replace $v$ with $v'$ and $v''$ at distance $\varepsilon > 0$, and then scale all distances by $\frac{1}{\varepsilon}$.}.

The graph representation of the metric $\calM$ is an undirected complete graph $(V,E)$, where each edge $e=(u,v)\in E$ has cost $\cost(e) = \dist_{\calM}(u,v)$. Without ambiguity, we use $\calM$ to denote this original graph and we call edges in $E$ \emph{original edges} (to differentiate them from \emph{virtual edges} that we introduce later). A feasible solution to the Steiner forest instance $(\calM,D)$ is a subset of original edges $F\subseteq E$ such that each demand pair $(u,v)\in D$ has $u$ and $v$ connected in the subgraph $(V,F)$. We do not require a feasible solution $F$ to be acyclic (a forest), but note that an optimal solution must be a forest.

In the online setting, the instance $(\calM,D)$ is given in an online fashion. At each moment (called \emph{arrival}) $1\leq t\leq n$, a demand pair $(u^{(t)},v^{(t)})\in D$ arrives. Moreover, $\dist_{\calM}(u^{(t)},v^{(t)})$ and $\dist_{\calM}(u^{(t)},x),\dist_{\calM}(v^{(t)},x)$ for all arrived terminals $x$ are revealed. In other words, if we let $D^{(t)}$ denote the set of arrived demand pairs and $V^{(t)}$ denote the arrived terminals, then after this arrival, we only know the submetric $\calM^{(t)}$ of $\calM$ induced by $V^{(t)}$. 

For the online instance $\{(\calM^{(t)},D^{(t)})\mid 1\leq t\leq n\}$, an online solution $\{F^{(t)}\mid 1\leq t\leq n\}$ has competitive ratio $\alpha$ and amortized recourse $\delta$ if each $F^{(t)}$ is an $\alpha$-approximate solution to the instance $(\calM^{(t)},D^{(t)})$, and the total number of edge insertions and deletions required to maintain the online solution across all $n$ arrivals is at most $n\delta$.

\begin{remark}
\label{remark:model}
The way we define a Steiner forest instance $(\calM,D)$ may seem unusual in the sense that $\calM$ is a metric only on the set of terminals (without Steiner vertices). However, we point out that this model is in fact common in the literature on online Steiner trees with low recourse \cite{MegowSVW12,GuGK13,gupta2014online}. The motivation for this definition is that maintaining a Steiner tree/forest in a general (non-metric) graph with low recourse is unrealistic, since even purchasing a simple path between two terminals can already incur large recourse.\footnote{In the seminal work \cite{imase1991dynamic} on online Steiner trees with recourse, the authors consider a different model in which recourse is measured by the number of \emph{primitive path} insertions and deletions, allowing them to work with general graphs.} Moreover, restricting to the \emph{terminal-only} metric increases the solution cost by at most a factor of 2 compared to the general model, so this model still captures the Steiner forest problem well, particularly when we aim for an $O(1)$-approximation.
\end{remark}

\begin{remark}
We assume that, at the beginning, we know [a constant-estimation $\hat{n}$ of] the length $n$ of the online instance $\{(\calM^{(t)},D^{(t)})\mid 1\leq t\leq n\}$. This assumption can be easily removed as follows. Initially, set the estimation $\hat{n}$ to be a constant. Whenever the current number of arrivals exceeds $\hat{n}$, we double $\hat{n}$ and restart the entire algorithm. Note that these restarts will only increase the final amortized recourse by a constant factor.
\end{remark}

\paragraph{Clusterings.} Given a terminal set $V$, a \emph{clustering} $\scrC$ is a partitioning of $V$ into \emph{clusters}. The \emph{trivial clustering} is the one in which each terminal forms its own singleton cluster. For two clusterings $\scrC_{1}$ and $\scrC_{2}$ (they can be clusterings of different terminal sets $V_{1}$ and $V_{2}$), we write $\scrC_{1}\preceq \scrC_{2}$ if each cluster $C_{1}\in \scrC_{1}$ is contained in some cluster $C_{2}\in \scrC_{2}$, i.e., $C_1 \subseteq C_2$. Note that if $\scrC_{1}\preceq \scrC_{2}$, there must be $V_{1}\subseteq V_{2}$.

\paragraph{Contracted Metrics.} Given an original metric $\calM = (V,\dist_{\calM})$ and a clustering $\scrC$ of $V$, consider the graph obtained by contracting each cluster $C\in\scrC$ in the graph representation of the original metric $\calM$ into a single vertex representing $C$. Note that the vertex set of this graph is exactly $\scrC$. We refer to the shortest-path metric of this graph as $\calM/\scrC=(\scrC,\dist_{\calM/\scrC})$, and denote this graph by $G_{\orig}(\calM/\scrC)$.

Sometimes we will further contract the graph $G_{\orig}(\calM/\scrC)$ by a subset $E'\subseteq E$ of original edges. That is, consider edges $e\in E'$ one by one, and for each $e=(u,v)$, contract the two vertices corresponding to $u,v$ into a single vertex (it is possible that $u,v$ correspond to the same vertex in the current graph, in which case we do nothing). We denote the resulting graph by $G_{\orig}((\calM/\scrC)/E')$, and use $\dist_{(\calM/\scrC)/E'}(\cdot,\cdot)$ to denote its shortest-path metric. The vertex set of $G_{\orig}((\calM/\scrC)/E')$ naturally corresponds to a partition of the vertex set of $G_{\orig}(\calM/\scrC)$ (which is exactly $\scrC$), so for $C_{1},C_{2}\in\scrC$, $\dist_{(\calM/\scrC)/E'}(C_{1},C_{2})$ can be defined naturally and it is unambiguous to talk about a $C_{1}$-$C_{2}$ shortest path in $G_{\orig}((\calM/\scrC)/E')$.

\section{A Clustering Procedure}
\label{sect:Clustering}

In this section, we describe a clustering procedure  which is the core part of the offline constant-approximate Steiner forest algorithm (called the \emph{timed gluttonous algorithm}) in \cite{DBLP:conf/stoc/Gupta015}.

The input of the clustering procedure is a Steiner forest instance $(\calM,D)$, and it outputs a \emph{clustering hierarchy} $\mathscr{H} = (\scrC_{0},\scrC_{1},\scrC_{2},...,\scrC_{L+1})$ of the terminal set $V$, where $L$ denotes the maximum level (which will be defined shortly). The hierarchy $\mathscr{H}$ satisfies that $\scrC_{0}$ is the trivial clustering of $V$, and for each $0\leq i\leq L$, $\scrC_{i}\preceq \scrC_{i+1}$. %

\paragraph{The Clustering Procedure.} Initially, for each terminal $v\in V$, define its \emph{level} to be
\[
\level(v) = \lceil \log_{2}\dist_{\calM}(v,u)\rceil
\]
where $u$ is the mate of $v$. Note that $\level(v)\geq 0$ since we have assumed $\dist_{\calM}(\cdot,\cdot)$ takes values at least $1$. The level of a cluster $C\subseteq V$ is $\level(C) = \max_{v\in C}\level(v)$.
Let $L = \max_{v\in V}\level(v)$ be the maximum level.

Then we iterate $i$ from $0$ to $L$. For each iteration $i$, we will construct $\scrC_{i+1}$ based on $\scrC_{i}$ as follows. 

\begin{enumerate}
\item\label{item:virtual} First, we construct an auxiliary graph $H_{i}$, called the \emph{virtual graph}, with vertices 
\[
V(H_{i}) = \{C\in \scrC_{i}\mid \level(C)\geq i\}
\]
corresponding to clusters $C\in \scrC_{i}$ with $\level(C)\geq i$, called \emph{$i$-active} clusters. Naturally, each $C\in\scrC_{i}$ with $\level(C)<i$ is an \emph{$i$-inactive} cluster. For each $C_{1},C_{2}\in V(H_{i})$, there is an edge (called a \emph{virtual edge}) in $H_{i}$ connecting them iff 
\[
\dist_{\calM / \scrC_{i}}(C_{1},C_{2})< 2^{i+1}.
\]
\item\label{item:Merging} We construct  $\scrC_{i+1}$ by \emph{contracting $H_{i}$ over $\scrC_{i}$}. That is, we first copy all $i$-inactive clusters from $\scrC_{i}$ into $\scrC_{i+1}$. Next, 
for each connected component $Q$ of $H_{i}$, 
we add a cluster $C_{i+1}$ into $\scrC_{i+1}$ which is the union of all $\scrC_{i}$-clusters in this connected component, i.e., $C_{i+1} = \bigcup_{C_{i}\in Q} C_{i}$ (note that the $C_{i+1}$ added in this way are $i$-active). Note that $\scrC_{i}\preceq \scrC_{i+1}$ by the construction.
\end{enumerate}

We note that the above clustering procedure is identical to that in Section 4.1 of \cite{DBLP:conf/stoc/Gupta015}, so we list some observations below without formal proofs. These observations will be useful when analyzing our online algorithm in \Cref{sec:our-alg}.

\begin{observation}
\label{lemma:ClustersGap}
For each $i\geq 0$ and any two different $i$-active clusters $C_{1},C_{2}\in \scrC_{i}$, we have $\dist_{\calM/\scrC_{i}}(C_{1},C_{2})\geq 2^{i}$.
\end{observation}
The above \Cref{lemma:ClustersGap} follows
because in iteration $i-1$, we merge two $(i-1)$-active clusters in $\scrC_{i-1}$ if they are close (i.e., they have distance less than $2^{i}$ in $\calM/\scrC_{i-1}$). A formal proof can be found in Section 4 in \cite{DBLP:conf/stoc/Gupta015}.

\begin{observation}
\label{ob:TopClustering}
For each demand pair $(u,v)\in D$, there is a cluster $C\in\scrC_{L+1}$ such that $u,v\in C$.
\end{observation}
This is basically because $\dist_{\calM}(u,v)<2^{\level(u)+1}$ by the definitions of $\level(\cdot)$ and $L$. Then at the beginning of iteration $i=\level(u)$, if $u$ and $v$ are still inside different clusters $C_{u},C_{v}\in \scrC_{i}$, then $C_{u}$ and $C_{v}$ are still $i$-active and they will be merged in this iteration.

\paragraph{The Forest-Forming Procedure.} For better understanding, we briefly describe a procedure forming a Steiner forest solution based on the clustering hierarchy (which is actually straightforward). However, we are not going to formally show the feasibility and approximation of this procedure, since the forest-forming procedure in our online algorithm is specialized. 

We initialize an empty solution $F$. For each $i$ from $0$ to $L$, we select an arbitrary \emph{virtual spanning forest} $\hat{F}_{i}$ of the virtual graph $H_{i}$, and then for each virtual edge $\hat{e}\in \hat{F}_{i}$ connecting two clusters $C_{1},C_{2}\in\scrC_{i}$, add into $F$ \underline{the original edges on} a $C_{1}$-$C_{2}$ shortest path $P$ in the graph $G_{\orig}(\calM/\scrC_{i})$. Note that the path $P$ has cost at most $2^{i+1}$, as guaranteed by the definition of virtual edges in Step~\ref{item:virtual}.

The feasibility of the solution $F$ basically comes from \Cref{ob:TopClustering}, and the approximation is given by the following \Cref{lemma:OfflineCost}, which will also be useful when analyzing the approximation of our online algorithm in \Cref{sec:our-alg}. \Cref{lemma:OfflineCost} is proved almost explicitly\footnote{The $m_{i}$ in their proof is exactly $|\hat{F}_{i}|$ in our context. They show that $\sum_{i}m_{i}2^{i+1}\leq O(1)\cdot\OPT$. Combining this with $|\hat{F}_{i}| = |\scrC_{i}| - |\scrC_{i+1}|$ gives \Cref{lemma:OfflineCost}.} in the proof of Theorem 4.2 in \cite{DBLP:conf/stoc/Gupta015}.

\begin{lemma}[\cite{DBLP:conf/stoc/Gupta015}]
\label{lemma:OfflineCost}
Let $\OPT$ be the cost of an optimal solution to the instance $(\calM,D)$. Then
\[
\sum_{i=0}^{L} (|\scrC_{i}|-|\scrC_{i+1}|)\cdot 2^{i+1}\leq O(1)\cdot\OPT \,.
\]
\end{lemma}

Let us explain \Cref{lemma:OfflineCost} somewhat further. The left-hand side takes a summation over all levels $i$. For each level $i$, observe that  $|\scrC_{i}| - |\scrC_{i+1}| = |\hat{F}_{i}|$. Hence the inequality basically says that if we assign a budget of $2^{i+1}$ to each virtual edge in $\hat{F}_{i}$, then the total budget is within a constant factor of $\OPT$. Indeed, the cost of the original edges added into the real solution $F$ by each virtual edge is within the budget.

\section{Our Online Algorithm} \label{sec:our-alg}

In this section, we will present our online Steiner forest algorithm with low recourse, and prove \Cref{thm:FullMain}. 

\begin{theorem}
\label{thm:FullMain}
    There is an algorithm for online Steiner forest
    that, over the course of $n$ arrivals
    of terminal pairs $(u_i,v_i)$
    in a metric space,
    maintains an $O(\log(n)/\lambda)$-competitive solution
    while inserting and deleting at most $O(n \lambda)$ edges, for any $\lambda\in[1,\log n]$.
\end{theorem}

Throughout this section, $\lambda$ is the tradeoff parameter in \Cref{thm:FullMain}, and \Cref{thm:main} follows \Cref{thm:FullMain} by setting $\lambda = \log n$.

We start with \Cref{sect:ApplyingClusteringProcedure}, in which we will apply the clustering procedure from \Cref{sect:Clustering} for each arrival and introduce some notation and observations regarding the clustering hierarchies. Next, in \Cref{sect:OnlineSolution}, we define the online Steiner forest solution $F$ algorithmically by, for each arrival $t$, constructing the \emph{snapshot} $F^{(t)}$ of $F$ based on the current clustering hierarchy. Finally, \Cref{thm:FullMain} follows from the analyses of feasibility, recourse, and approximation in \Cref{sect:Feasibility}, \Cref{sect:Recourse}, and \Cref{sect:Approximation}, respectively. 

For convenience, we provide an outline of the full online algorithm as \Cref{algo:main} (page~\pageref{algo:main}). The details of each step are still described in the text of \Cref{sect:OnlineSolution}.

\subsection{Applying the Clustering Procedure}
\label{sect:ApplyingClusteringProcedure}

We start with some notation. Let $\{(\calM^{(t)},D^{(t)})\mid 1\leq t\leq n\}$ be the given online instance. For each arrival $t$, we run the clustering procedure for the (offline) instance $(\calM^{(t)},D^{(t)})$.
Here we will write the corresponding variables $L$ (the maximum level), $H_{i}$ (the virtual graph at level $i$), $\scrC_{i}$ (the clustering at level $i$),
and $\mathscr{H}$ (the clustering hierarchy)
with a superscript $(t)$. To avoid clutter, for each arrival $t\geq 1$ and $i\geq L^{(t)}+1$, we let $\scrC^{(t)}_{i}$ be the same as the top-level clustering $\scrC^{(t)}_{L^{(t)}+1}$, and let $H^{(t)}_{i}$ be an empty virtual graph. Furthermore, for $t=0$, we define $L^{(0)} = 0$ and each $\scrC^{(0)}_{i}$ and $H^{(0)}_{i}$ to be empty. The \Cref{ob:ClusteringRelationSameLevel} below will be useful later.

\begin{lemma}
\label{ob:ClusteringRelationSameLevel}
For each $t\geq 1$ and $0\leq i\leq L^{(t)}+1$, we have $\scrC^{(t-1)}_{i}\preceq \scrC^{(t)}_{i}$.
\end{lemma}
\begin{proof}
    We prove this by induction on $i$. Since $\scrC^{(t-1)}_{0}$ is the trivial clustering on the terminal set $V^{(t - 1)}$ and $\scrC^{(t)}_{0}$ is the trivial clustering on the terminal set $V^{(t)}$, the base case is true. Now assume for some fixed $i$, $\scrC^{(t-1)}_{i}\preceq \scrC^{(t)}_{i}$. We will show $\scrC^{(t-1)}_{i + 1}\preceq \scrC^{(t)}_{i + 1}$. In other words, our goal is to show that a cluster $C \in \scrC^{(t-1)}_{i + 1}$ is fully contained in some cluster of $\scrC^{(t)}_{i + 1}$. To this end, we consider the connected components of $H_i^{(t-1)}$ and the connected components of $H_i^{(t)}$. Note that if some two clusters $C_1, C_2 \in \scrC^{(t-1)}_{i}$ have 
    \[ \dist_{\calM / \scrC^{(t - 1)}_{i}}(C_{1},C_{2})< 2^{i+1}\]
    then, since $\scrC^{(t-1)}_{i}\preceq \scrC^{(t)}_{i}$, we also have that for the parent clusters $D_1, D_2 \in \scrC^{(t)}_{i}$ which contain $C_1$ and $C_2$ respectively,
    \[ \dist_{\calM / \scrC^{(t)}_{i}}(D_{1},D_{2})< 2^{i+1}.\]
    Therefore, an edge between $C_1$ and $C_2$ in $H_i^{(t-1)}$ is also an edge between $D_1$ and $D_2$ in $H_i^{(t)}$. So, if a cluster $C \in \scrC^{(t-1)}_{i + 1}$ results from some connected component $Q$ in $H_i^{(t-1)}$ consisting of $C_1, \hdots, C_k$, then there is a corresponding connected component $Q'$ of $D_1, \hdots, D_k$ in $H_i^{(t)}$ (note that the $D_\ell$ are not necessarily distinct). This proves our inductive step. 
\end{proof}

\subsection{Constructing the Snapshots of the Online Solution}
\label{sect:OnlineSolution}

Consider an arrival $1\leq t\leq n$. We construct the snapshot $F^{(t)}$ of the online solution in two phases.

\paragraph{The First Phase.} We first construct a \emph{virtual solution} $\hat{F}^{(t)}$ using virtual edges in the virtual graphs $H^{(t)}_{i}$ for all $i$. Similarly to the forest-forming procedure in \Cref{sect:Clustering}, for each level $0\leq i\leq L^{(t)}$, we will pick a virtual spanning forest $\hat{F}^{(t)}_{i}$ of $H^{(t)}_{i}$. However, rather than choosing an arbitrary forest as $\hat{F}^{(t)}_{i}$, we will choose $\hat{F}^{(t)}_{i}$ more carefully, as we discuss shortly. The virtual solution $\hat{F}^{(t)} = \bigcup_{0\leq i\leq L^{(i)}} \hat{F}^{(t)}_{i}$ is simply the union of virtual spanning forests at all levels.

\medskip

\noindent\underline{Inheritable and Inherited Virtual Edges.} Before describing how to choose $\hat{F}^{(i)}_{t}$, we need to introduce the concepts of \emph{inheritable/non-inheritable} and \emph{inherited/non-inherited} virtual edges.

For each virtual edge $\hat{e}^{(t-1)} = (C_{1},C_{2})\in \hat{F}^{(t-1)}_{i}$,
it is \emph{inheritable} if $C_{1},C_{2}\in \scrC^{(t-1)}_{i}$ are not contained in the same $\scrC^{(t)}_{i}$-cluster (recall that $\scrC^{(t-1)}_{i}\preceq \scrC^{(t)}_{i}$ from \Cref{ob:ClusteringRelationSameLevel}), otherwise it is \emph{non-inheritable}.

\begin{lemma}
\label{lemma:VirtualEdgeMapping}
For each inheritable virtual edge $\hat{e}^{(t-1)} = (C_{1},C_{2})\in \hat{F}^{(t-1)}_{i}$, there is a virtual edge $\hat{e}^{(t)}\in H^{(t)}_{i}$ connecting $D_1, D_2 \in \scrC^{(t)}_{i}$, the two different $\scrC^{(t)}_{i}$-clusters containing $C_{1},C_{2}$ respectively. 
\end{lemma}
\begin{proof}
    In the same vein as the proof above, note that if $\hat{e}^{(t-1)} = (C_{1},C_{2})$ is an inheritable virtual edge in $H^{(t-1)}_{i}$, then: (1) $C_{1},C_{2}\in \scrC^{(t-1)}_{i}$ are contained in distinct clusters $D_1$ and $D_2$ in $\scrC^{(t)}_{i}$, and (2) $\dist_{\calM / \scrC^{(t - 1)}_{i}}(C_{1},C_{2})< 2^{i+1}$. The second property implies $\dist_{\calM / \scrC^{(t)}_{i}}(D_{1},D_{2})< 2^{i+1}$, and therefore we will have an edge $\hat{e}^{(t)}\in H^{(t)}_{i}$ connecting $D_1$ and $D_2$.
\end{proof}

\Cref{lemma:VirtualEdgeMapping} naturally defines a mapping $\pi^{(t)}_{i}$ from the inheritable virtual edges in $\hat{F}^{(t-1)}_{i}$ to virtual edges in $H^{(t)}_{i}$. We define the \emph{inherited} virtual edges in $H^{(t)}_{i}$ to be the image set of $\pi^{(t)}_{i}$, and will call the other virtual edges in $H^{(t)}_{i}$ \emph{non-inherited}. For each inherited virtual edge $\hat{e}^{(t)}\in H^{(t)}_{i}$, we fix an arbitrary inheritable virtual edge $\hat{e}^{(t-1)}\in \hat{F}^{(t-1)}_{i}$ with $\pi^{(t)}_{i}(\hat{e}^{(t-1)}) = \hat{e}^{(t)}$ as the \emph{parent} of $\hat{e}^{(t)}$, and say that $\hat{e}^{(t)}$ is inherited from its parent $\hat{e}^{(t-1)}$.

\medskip

\noindent\underline{Choosing $\hat{F}^{(i)}_{t}$.} Now, we pick the virtual spanning forest $\hat{F}^{(t)}_{i}$ giving priority to the inherited virtual edges in $H^{(t)}_{i}$. Formally speaking, we first pick an arbitrary spanning forest $\hat{F}^{(t)}_{\inh,i}$ in $H^{(t)}_{\inh,i}$, the subgraph of $H^{(t)}_{i}$ induced by the inherited virtual edges. Then we arbitrarily augment $\hat{F}^{(t)}_{\inh,i}$ to be 
a spanning forest $\hat{F}^{(t)}_{i}$ of $H^{(t)}_{i}$.

For better understanding, we emphasize that virtual edges in $\hat{F}^{(t)}_{i}$ are classified in two ways: inheritable vs.~non-inheritable, and inherited ($\hat{F}^{(t)}_{\inh,i}$) vs.~non-inherited ($\hat{F}^{(t)}_{i}\setminus \hat{F}^{(t)}_{\inh,i}$). Each inherited virtual edge in $\hat{F}^{(t)}_{i}$ is inherited from its parent, some inheritable virtual edge in $\hat{F}^{(t-1)}_{i}$. However, not every inheritable virtual edge in $\hat{F}^{(t-1)}_{i}$ may serve as the parent of an inherited virtual edge in $\hat{F}^{(t)}_{i}$.

Note that $\scrC^{(t)}_{i+1}$ is exactly the clustering obtained by contracting $\hat{F}^{(t)}_{i}$ over $\scrC^{(t)}_{i}$ (the meaning of contracting is the same as in Step \ref{item:Merging} of the clustering procedure), and we let $\scrC^{(t)}_{\inh,i}$ denote the clustering obtained by contracting $\hat{F}_{\inh,i}$ over $\scrC^{(t)}_{i}$. Then we have the following \Cref{lemma:ClusteringRelationInherit} which will be useful later.

\begin{lemma}
\label{lemma:ClusteringRelationInherit}
$\scrC^{(t-1)}_{i+1}\preceq \scrC^{(t)}_{\inh,i}$.
\end{lemma}
\begin{proof}
Assume for contradiction that there exists a cluster $C^{(t-1)}_{i+1}\in \scrC^{(t-1)}_{i+1}$ such that $C^{(t-1)}_{i+1}$ intersects two different clusters in $\scrC^{(t)}_{\inh,i}$.

First, we claim that $C^{(t-1)}_{i+1}$ is $i$-active. Assume the opposite. We must have $C^{(t-1)}_{i+1}\in \scrC^{(t-1)}_{i}$ by the description of the clustering procedure. By \Cref{ob:ClusteringRelationSameLevel}, we know 
\[
\scrC^{(t-1)}_{i}\preceq \scrC^{(t)}_{i}\preceq \scrC^{(t)}_{\inh,i}.
\]
Hence $C^{(t-1)}_{i+1}$ is contained in some cluster in $\scrC^{(t)}_{\inh,i}$, a contradiction.  

Given that $C^{(t-1)}_{i+1}$ is $i$-active, it is the union of several $i$-active clusters in $\scrC^{(t-1)}_{i}$ by the clustering procedure. Then, by $\scrC^{(t-1)}_{i}\preceq \scrC^{(t)}_{\inh,i}$ and the assumption that $C^{(t-1)}_{i+1}$ intersects two different clusters in $\scrC^{(t)}_{\inh,i}$, there must be two different $i$-active $\scrC^{(t-1)}_{i}$-clusters $C^{(t-1)}_{i,1}$ and $C^{(t-1)}_{i,2}$ inside $C^{(t-1)}_{i+1}$ such that they are contained by different $\scrC^{(t)}_{\inh,i}$-clusters.

Let $C^{(t)}_{i,1}$ and $C^{(t)}_{i,2}$ be the $\scrC^{(t)}_{i}$-clusters containing $C^{(t-1)}_{i,1}$ and $C^{(t-1)}_{i,2}$ respectively (we use $\scrC^{(t-1)}_{i}\preceq \scrC^{(t)}_{i}$ from \Cref{ob:ClusteringRelationSameLevel} again). If $C^{(t)}_{i,1}$ and $C^{(t)}_{i,2}$ are the same, combining it with $\scrC^{(t)}_{i}\preceq \scrC^{(t)}_{\inh,i}$ already leads to a contradiction. Hence we assume $C^{(t)}_{i,1}$ and $C^{(t)}_{i,2}$ are different clusters. Now, we have already known that $C^{(t-1)}_{i,1}$ and $C^{(t-1)}_{i,2}$ are in the same connected component of $H^{(t-1)}_{i}$ (i.e., they are connected by the virtual spanning forest $\hat{F}^{(t-1)}_{i}$). From the inheritance relation between $\hat{F}^{(t-1)}_{i}$-virtual edges and $H^{(t)}_{i}$-virtual edges, it is straightforward to see that $C^{(t)}_{i,1}$ and $C^{(t)}_{i,2}$ are in the same connected component of $H^{(t)}_{\inh,i}$, which means they are contained by the same $\scrC^{(t)}_{\inh,i}$-cluster, a contradiction.

\end{proof}

\paragraph{The Second Phase.} In the second phase, we will construct the snapshot $F^{(t)}$ based on the virtual solution $\hat{F}^{(t)}$. A natural attempt is to proceed analogously to the offline forest-forming procedure in \Cref{sect:Clustering}: for each virtual edge $\hat{e} = (C_{1},C_{2})\in \hat{F}^{(t)}_{i}\subseteq \hat{F}^{(t)}$, add into $F^{(t)}$ a $C_{1}$-$C_{2}$ shortest path in the metric $\calM / \scrC^{(t)}_{i}$. This will give a feasible $F^{(t)}$ with good approximation as we discussed in \Cref{sect:Clustering}. However, it gives no control over the recourse: such a shortest path may be made up of many original edges in the original metric $\calM$. 

We fix this issue by \emph{pinning} edges. Roughly speaking, whenever we add a shortest path, we pin some fraction of its original edges. Pinning an original edge means that we will never remove it from our online solution. We will guarantee two properties of the pinned edges. First, the total cost of pinned edges is always within a constant factor of the optimum, so that the online solution can achieve a constant competitive ratio. Second, the number of pinned edges grows linearly in $t$, so that we can achieve low recourse by charging the number of original edge insertions to the number of pinned edges. We now formalize this argument.

As demands arrive, we will maintain a growing set $A$ of pinned edges. Our job in the current arrival $t$ is to assign to each virtual edge $\hat{e}^{(t)}\in \hat{F}^{(t)}$ a set of original edges in $\calM$, denoted by $E_{\orig}(\hat{e}^{(t)})$.
From the previous arrival $t-1$, we are already given the pinned edge set $A^{(t-1)}$ and the original edge set $E_{\orig}(\hat{e}^{(t-1)})$ of every virtual edge $\hat{e}^{(t-1)}\in \hat{F}^{(t-1)}$. Meanwhile, we will also update the pinned edge set from $A^{(t-1)}$ to $A^{(t)}$.
The original solution is then
\[
F^{(t)} = A^{(t)}\cup \bigcup_{\hat{e}^{(t)}\in \hat{F}^{(t)}} E_{\orig}(\hat{e}^{(t)}).
\]

For each inherited virtual edge $\hat{e}^{(t)}\in \hat{F}^{(t)}$, its original  edge set is simply
\[
E_{\orig}(\hat{e}^{(t)}) = E_{\orig}(\hat{e}^{(t-1)})
\]
where $\hat{e}^{(t-1)}\in \hat{F}^{(t-1)}$ is the parent of $\hat{e}^{(t)}$
As for pinning, on a high level, we add a $1/\lambda$ fraction of the original edges to $A^{(t)}$ (selecting the cheapest ones);
if there are fewer than $\lambda$ many, we use a buffer set.
More precisely,
the original edge sets of the non-inherited virtual edges and the update from $A^{(t-1)}$ to $A^{(t)}$ are given by the following algorithm.

\begin{algorithm}[H]
\caption{Pinning}
\label{algo:PickOriginalEdgeSets}
\begin{algorithmic}[1]
\State Initialize $A^{(t)}\gets A^{(t-1)}$
\State Initialize $B\gets \emptyset$ to be a \emph{multiset}\Comment{initialize the buffer}
\For{each level $i$ and non-inherited $\hat{e}^{(t)} = (C_{1},C_{2})\in \hat{F}^{(t)}_{i}$}
\State\label{line:SP}$P\gets$ a $C_{1}$-$C_{2}$ shortest path in $G_{\orig}((\calM / \scrC^{(t)}_{i})/ A^{(t)})$
\State\label{line:Eorig}$E_{\orig}(\hat{e}^{(t)})\gets$ the original edges on $P$
\If {$|E_{\orig}(\hat{e}^{(t)})|\geq \lambda$}
\State\label{line:BatchPin}Add $\lfloor |E_{\orig}(\hat{e}^{(t)})|/\lambda\rfloor$ cheapest edges from $E_{\orig}(\hat{e}^{(t)})$ into $A^{(t)}$ 
\State $B\gets\emptyset$
\Else
\State\label{line:UpdateB}$B\gets B\uplus E_{\orig}(\hat{e}^{(t)})$  \Comment{multiset union}
\If{$|B|\geq \lambda$\label{cond:Bsize}} 
\State\label{line:SinglePin}Add the cheapest edge from $B$ into $A^{(t)}$
\State $B\gets\emptyset$
\EndIf
\EndIf
\EndFor
\end{algorithmic}
\end{algorithm}
We note that we define the buffer $B$ to be a multiset just for ease of analysis. Defining $B$ as a set instead will not affect the correctness.

Finally, to enhance understanding, we outline our entire online algorithm in \Cref{algo:main}.

\begin{algorithm}
\caption{Online Low-Recourse Steiner Forest}
\label{algo:main}
\begin{algorithmic}[1]
\For{each arrival $t$ with demand pair $(u^{(t)},v^{(t)})$} 
\State Compute the clustering hierarchy $\mathscr{H}^{(t)} = \{\scrC_{0}^{(t)},\scrC_{1}^{(t)},\scrC_{1}^{(t)},...\}$ \Comment{see \Cref{sect:Clustering}}
\For{each level $i$}
\State Compute $\hat{F}^{(t)}_{i}$, a set of virtual edges (between clusters in $\scrC^{(t)}_{i}$), each of which is either inherited or non-inherited \Comment{the first phase}
\EndFor
\State The virtual solution is $\hat{F}^{(t)} := \bigcup_{i} \hat{F}^{(t)}_{i}$
\State Update the pinned edge set $A^{(t)}$ and obtain the mapping $E_{\orig}$ from virtual edges $\hat{F}^{(t)}$ to sets of original edges, via \Cref{algo:PickOriginalEdgeSets} \Comment{the second phase}
\State Output online solution $F^{(t)}:=A^{(t)}\cup\bigcup_{\hat{e}^{(t)}\in \hat{F}^{(t)}} E_{\orig}(\hat{e}^{(t)})$
\EndFor
\end{algorithmic}
\end{algorithm}

\subsection{Feasibility}
\label{sect:Feasibility}

To prove feasibility, it suffices to show that virtual edges between two clusters $C_1$ and $C_2$ do indeed correspond to real edge sets which connect $C_1$ and $C_2$.

\begin{lemma}
\label{lemma:Feasibility}
For each arrival $t$, level $0\leq i\leq L^{(t)}$, and virtual edge $\hat{e}^{(t)} = (C^{(t)}_{1},C^{(t)}_{2})\in\hat{F}^{(t)}_{i}$, the clusters $C^{(t)}_{1}$ and $C^{(t)}_{2}$ belong to the same vertex in the graph $G_{\orig}((\calM / \scrC^{(t)}_{i}) / (A^{(t)}\cup E_{\orig}(\hat{e}^{(t)})))$.
\end{lemma}
\begin{proof}
If $\hat{e}^{(t)}$ is a non-inherited virtual edge, this statement directly follows from \Cref{line:SP,line:Eorig} of \Cref{algo:PickOriginalEdgeSets}. Concretely, at the moment $E_{\orig}$ is defined at \Cref{line:Eorig}, $C^{(t)}_{1}$ and $C^{(t)}_{2}$ clearly belong to the same vertex in $(\calM / \scrC^{(t)}_{i}) / (A^{(t)}\cup E_{\orig}(\hat{e}^{(t)}))$ by definition. Therefore, at the end of arrival $t$, since $A^{(t)}$ only grows during \Cref{algo:PickOriginalEdgeSets}, the statement also holds.

The argument for an inherited virtual edge $\hat{e}^{(t)}$ is similar, but we need induction to formally prove it. Assume inductively that the statement of \Cref{lemma:Feasibility} holds for all virtual edges in $\hat{F}^{(t-1)}$. Let $\hat{e}^{(t-1)} = (C^{(t-1)}_{1},C^{(t-1)}_{2})\in\hat{F}^{(t-1)}_{i}$ be the parent of $\hat{e}^{(t)}$, which means $C^{(t)}_{1}\supseteq C^{(t-1)}_{1}$, $C^{(t)}_{2}\supseteq C^{(t-1)}_{2}$ and $E_{\orig}(\hat{e}^{(t)}) = E_{\orig}(\hat{e}^{(t-1)})$. Then the statement holds for $\hat{e}^{(t)}$ by the induction hypothesis and because $\scrC^{(t-1)}_{i}\preceq \scrC^{(t)}_{i}$ %
\end{proof}

And so, since $\hat{F}_i^{(t)}$ virtually connects up demands in the arrived demand set (by \Cref{ob:TopClustering}), we know that $F_i^{(t)}$ will be a feasible solution. 

\subsection{Recourse}
\label{sect:Recourse}

We are going to bound the total number of edge insertions by $O(n\lambda)$, that is,
\[
\sum_{t=1}^{n}|F^{(t)}\setminus F^{(t-1)}|\leq O(n\lambda),
\]
which gives $O(\lambda)$ amortized recourse (since the number of edge deletions does not exceed the number of edge insertions).

Consider the changes from $F^{(t-1)}$ to $F^{(t)}$. The new edges $F^{(t)}\setminus F^{(t-1)}$ come from either $A^{(t)}\setminus A^{(t-1)}$ or $\bigcup_{\text{non-inherited }\hat{e}^{(t)}\in \hat{F}^{(t)}} E_{\orig}(\hat{e}^{(t)})$ (since each inherited $\hat{e}^{(t)}\in \hat{F}^{(t)}$ has $E_{\orig}(\hat{e}^{(t)})\subseteq \hat{F}^{(t-1)}$). In other words, we have
\[
\sum_{t=1}^{n} |F^{(t)}\setminus F^{(t-1)}|\leq \sum_{t=1}^{n}\left(|A^{(t)}\setminus A^{(t-1)}| + \sum_{\text{non-inherited }\hat{e}^{(t)}\in \hat{F}^{(t)}} |E_{\orig}(\hat{e}^{(t)})|\right).
\]

For the pinned edge set, we have $\sum_{t=1}^{n}|A^{(t)}\setminus A^{(t-1)}|\leq 2n-1$ by the fact that $A$ only grows and by the following \Cref{lemma:PinnedEdgeSize}.

\begin{lemma}
\label{lemma:PinnedEdgeSize}
There are at most $2n-1$ pinned edges, i.e., $|A^{(n)}|\leq 2n-1$.
\end{lemma}
\begin{proof}
Observe that whenever we pin several edges at \Cref{line:BatchPin} or one edge at \Cref{line:SinglePin}, the new pinned edges do not form a cycle with the existing pinned edges on the terminal set $V^{(t)}$. Hence at the end $A^{(n)}$ is a forest on $V^{(n)}$, which implies that $|A^{(n)}|\leq |V^{(n)}|-1= 2n-1$.
\end{proof}

Next, we charge the total size of original edge sets of non-inherited virtual edges (\emph{non-inherited original edge sets} for short) to the number of pinned edges. In every execution of \Cref{line:BatchPin}, we charge each new pinned edge a fee of $10\lambda$ units to account for $|B| + |E_{\orig}(\hat{e}^{(t)})|$, which is feasible because
\[
|B| + |E_{\orig}(\hat{e}^{(t)})|\leq 10\lambda\cdot \lfloor |E_{\orig}(\hat{e}^{(t)})|/\lambda\rfloor
\]
when $|E_{\orig}(\hat{e}^{(t)})|\geq \lambda$ and $|B|<\lambda$ (guaranteed by \Cref{cond:Bsize}). In every execution of \Cref{line:SinglePin}, again charge the new pinned edge a fee of $10\lambda$ units to account for $|B|$, which is available because $|B|\leq 2\lambda$ (right before \Cref{line:UpdateB}, the buffer $B$ has size less than $\lambda$, and the current $E_{\orig}(\hat{e}^{(t)})$ has size less than $\lambda$). Finally, observe that every non-inherited original edge set is accounted for in the charging argument (either directly or via the buffer $B$), except for those remaining in the buffer $B$ at the end of each arrival (at these moments, $|B|<\lambda$). Therefore, we have
\[
\sum_{t=1}^{n}\sum_{\text{non-inherited }\hat{e}^{(t)}\in \hat{F}^{(t)}}|E_{\orig}(\hat{e}^{(t)})|\leq 10\lambda\cdot |A^{(n)}| + n\lambda\leq O(n\lambda).
\]

\subsection{Approximation}
\label{sect:Approximation}

To avoid clutter, we only show that after the last arrival $n$, the online solution $F^{(n)}$ is an $O(1)$-approximation to the instance $(\calM^{(n)},D^{(n)})$. The approximation in the other arrivals can be bounded in exactly the same way.

Let $\OPT$ be the cost of the optimal solution to $(\calM^{(n)},D^{(n)})$. Recall that 
\[
F^{(n)} = A^{(n)}\cup \bigcup_{\hat{e}^{(n)}\in \hat{F}^{(n)}} E_{\orig}(\hat{e}^{(n)}) \,.
\]
Hence, to bound the cost of $F^{(n)}$ by $O(\log n\cdot\OPT/\lambda)$, it suffices to bound 
\begin{itemize}
\item $\cost(A^{(n)})$ (called the \emph{pinning cost}) by $O(\log n\cdot\OPT/\lambda)$, and 
\item $\sum_{\hat{e}^{(n)}\in \hat{F}^{(n)}}\cost(E_{\orig}(\hat{e}^{(n)}))$ (called the \emph{forest-forming cost}) by $O(\OPT)$.
\end{itemize}
\paragraph{The Forest-Forming Cost.} The forest-forming cost can be interpreted as the cost of the offline Steiner forest algorithm introduced in \Cref{sect:Clustering}, which can be directly bounded by \Cref{lemma:OfflineCost}.

Formally, for each level $i$ and each virtual edge $\hat{e}^{(n)}\in \hat{F}^{(n)}_{i}$, we have
\begin{equation}
\label{eq:OriginalEdgeSetCost}
\cost(E_{\orig}(\hat{e}^{(n)}))\leq 2^{i+1}.
\end{equation}
To see this, we trace $\hat{e}^{(n)}$ back through the inheritance relation until reaching a non-inherited virtual edge $\hat{e}^{(t)} = (C_{1},C_{2})\in \hat{F}^{(t)}_{i}$ in some arrival $t\leq n$. At the moment when $E_{\orig}(\hat{e}^{(t)})$ is defined on \Cref{line:SP,line:Eorig} of \Cref{algo:PickOriginalEdgeSets}, we have
\begin{equation}
\label{eq:OriginalEdgeSetCost2}
\cost(E_{\orig}(\hat{e}^{(t)}))\leq \dist_{(\calM / \scrC^{(t)}_{i})\setminus A^{(t)}}(C_{1},C_{2})\leq \dist_{\calM / \scrC^{(t)}_{i}}(C_{1},C_{2})\leq 2^{i+1},
\end{equation}
where the last inequality is by $\hat{F}^{(t)}_{i}\subseteq H^{(t)}_{i}$ 
and the definition of $H^{(t)}_{i}$. The inheritance relation implies $E_{\orig}(\hat{e}^{(n)}) = E_{\orig}(\hat{e}^{(t)})$, so we get the desired bound (\ref{eq:OriginalEdgeSetCost}).

Given (\ref{eq:OriginalEdgeSetCost}), we bound the forest-forming cost by
\begin{align*}
\sum_{\hat{e}^{(n)}\in \hat{F}^{(n)}}\cost(E_{\orig}(\hat{e}^{(n)}))\leq \sum_{i}|\hat{F}^{(n)}_{i}|\cdot 2^{i+1}= \sum_{i}(|\scrC^{(n)}_{i}|-|\scrC^{(n)}_{i+1}|)\cdot 2^{i+1}\leq O(\OPT),
\end{align*}
where the equality is by $|\hat{F}^{(n)}_{i}| = |\scrC^{(n)}_{i}|-|\scrC^{(n)}_{i+1}|$ and the last inequality is by \Cref{lemma:OfflineCost}.

\paragraph{The Pinning Cost.} \Cref{algo:PickOriginalEdgeSets}  pins roughly a $(1/\lambda)$-fraction of edges from the non-inherited original edge sets, so naturally the first step is to bound the total cost of non-inherited original edge sets.

\begin{restatable}{theorem}{ThmOnlineCost}
\label{lemma:OnlineCost}
For each level $i$ we have
\begin{equation}
\label{eq:OnlineCost}
\sum_{t}|\hat{F}^{(t)}_{i}\setminus \hat{F}^{(t)}_{\inh,i}|\cdot 2^{i+1}\leq O(\OPT) \,.
\end{equation}
\end{restatable}

We say that a non-inherited original edge set is from level $i$ if its corresponding non-inherited virtual edge is from $\hat{F}^{(t)}_{i}$ for some $t$.
\Cref{lemma:OnlineCost} above shows that for each level $i$, the total cost of non-inherited original edge sets summed over all arrivals is bounded by $O(\OPT)$ (note that each non-inherited original edge set has cost at most $2^{i+1}$ as we argued in (\ref{eq:OriginalEdgeSetCost2})). 

We prove \Cref{lemma:OnlineCost}
in \Cref{sec:proof_of_lemma_onlinecost} below.
Given \Cref{lemma:OnlineCost}, the intuition for bounding the pinning cost is as follows. If we assume that there are $O(\log n)$ levels, then \Cref{lemma:OnlineCost} shows that the total cost of non-inherited original edge sets is $O(\log n\cdot \OPT)$, so the pinning cost can be bounded by $O(\log n\cdot \OPT/\lambda)$ since we pinned roughly a $(1/\lambda)$ fraction. Without the assumption of $O(\log n)$ levels, the intuition is that non-inherited original edge sets not from the top $O(\log n)$ levels have negligible costs. We formalize this intuition as follows.

To see $\cost(A^{(n)})\leq O(\OPT)$, %
by \Cref{lemma:PinnedEdgeSize} it suffices to show $\cost(A')\leq O(\OPT)$, where $A' = \{e\mid e\in A^{(n)}\text{ s.t. }\cost(e)> \OPT/n\}$ are pinned edges with non-negligible costs. Note that $A$ will only grow at \Cref{line:BatchPin,line:SinglePin} in \Cref{algo:PickOriginalEdgeSets}. 

\medskip

\noindent\underline{\Cref{line:BatchPin}.} Consider an execution of \Cref{line:BatchPin} that brings new $A'$-pinned edges. The current non-inherited original edge set $E_{\orig}(\hat{e}^{(t)})$ must come from the top $\lceil\log n\rceil+2$ levels (i.e., the current level $i$ is between $L^{(n)} - \lceil\log n\rceil-1$ and the maximum level $L^{(n)}$). Otherwise, we have $\cost(E_{\orig}(\hat{e}^{(t)}))\leq 2^{i+1}\leq 2^{L^{(n)}-1}/n\leq \OPT/n$ (note that $2^{L^{(n)}-1}\leq \OPT$ by the definition of the function $\level$), and any pinned edge from $E_{\orig}(\hat{e}^{(t)})$ would have negligible cost. Furthermore, the algorithm guarantees that the total cost of new $A'$-pinned edges from this execution is at most $\cost(E_{\orig}(\hat{e}^{(t)}))/\lambda$.

\medskip

\noindent\underline{\Cref{line:SinglePin}.} Suppose an execution of \Cref{line:SinglePin} brings a new $A'$-pinned edge. %
Note that at this moment, the buffer $B$ is the \emph{multiset union} of several non-inherited original edge sets, all of which come from the top $\lceil\log n\rceil+2$ levels by the same argument as in the previous case. Again, the algorithm guarantees the cost of this new $A'$-pinned edge is at most $\cost(B)/\lambda$.

\medskip

Finally, observe that the sum of $\cost(E_{\orig}(\hat{e}^{(t)}))$ over such executions of \Cref{line:BatchPin} (which brings new $A'$ edges) and $\cost(B)$ over such executions of \Cref{line:SinglePin} is at most $O(\log n\cdot \OPT)$ by \Cref{lemma:OnlineCost} (since the involved non-inherited original edge sets are all from the top $\lceil \log n\rceil+2$ levels), so we get $\cost(A')\leq O(\log n\cdot\OPT/\lambda)$ and $\cost(A)\leq \cost(A') + O(\OPT)\leq O(\log n\cdot\OPT/\lambda)$.

\subsection{Proof of \Cref{lemma:OnlineCost}} \label{sec:proof_of_lemma_onlinecost}

We will invoke a fact that forms the initial part of the analysis in~\cite{DBLP:conf/stoc/Gupta015}. By losing a constant factor in cost, it allows us to assume that the clustering induced by our solution is a refinement of that of the optimal solution (in~\cite{DBLP:conf/stoc/Gupta015} this is called the faithfulness property). For clarity, we refer to $\scrC^{(n)}_{L^{(n)}+1}$ as $\scrC^{(n)}_{\max}$, which is the top-level clustering of the hierarchy $\mathscr{H}^{(n)}$ after the last arrival.

\begin{lemma}[Theorem 4.1 in \cite{DBLP:conf/stoc/Gupta015}]
\label{lemma:LPsolution}
There exists a solution $F^{\star\star}$ of the instance $(\calM, D)$ satisfying that
\begin{itemize}
\item the cost of $F^{\star\star}$ is at most $O(\OPT)$, and
\item for each cluster $C\in \scrC^{(n)}_{\max}$, 
all vertices in $C$ belong to the same connected component of $F^{\star\star}$.
\end{itemize}
\end{lemma}

Now we prove \Cref{lemma:OnlineCost} using \Cref{lemma:LPsolution}.
We restate \Cref{lemma:OnlineCost} for ease of reading.

\ThmOnlineCost*

\begin{proof}
We use a dual fitting argument.

\paragraph{Primal LP and its Dual.} Consider the following primal LP, where ${\cal S}$ is the family of all cuts $S\subseteq V$ separating at least one cluster in $\scrC^{(n)}_{\max}$, i.e., there exists a $C\in \scrC^{(n)}_{\max}$ with $C\cap S\neq \emptyset$ and $C\setminus S\neq \emptyset$. For each cut $S\in {\cal S}$, we use $\partial(S)\subseteq E$ to denote the set of original edges crossing $S$.
\begin{align*}
\min \qquad & \sum_{e\in E} \cost(e)\cdot x(e)
\\\text{s.t.} \qquad & \sum_{e\in\partial(S)} x(e)\geq 1~~~~~ \forall S\in {\cal S}
\\ & x\ge0.
\end{align*}
Note that this primal LP does not exactly correspond to the instance $(\calM, D)$, since it requires the solution to have stronger connectivity (each cluster in $\scrC^{(n)}_{\max}$ should have all its vertices connected). However, its optimal value $\OPT_{\LP}$ will not exceed $\OPT$ too much, i.e., we have
\[
\OPT_{\LP}\leq O(\OPT),
\]
because the solution $F^{\star\star}$ in \Cref{lemma:LPsolution} gives a feasible solution (for each $e\in E$, set $x(e) = 1$ if $e\in F^{\star\star}$, otherwise $x(e) = 0$) to the LP with cost at most $O(\OPT)$. The dual LP is as follows:
\begin{align*}
\max \qquad & \sum_{S\in {\cal S}} y(S)
\\\text{s.t.} \qquad & \sum_{S \; : \; e\in\partial(S)} y(S)\leq \cost(e)~~~~~ \forall e\in E
\\ & y\ge0.
\end{align*}

An intermediate goal is to construct a feasible dual solution with value approximately the LHS of (\ref{eq:OnlineCost}). We slightly abuse notation by viewing a dual solution as a function $y:2^{V}\to \mathbb{R}_{\geq 0}$, and we say such a dual solution $y$ is feasible if it satisfies the dual constraints and is zero outside the actual domain ${\cal S}$.

\paragraph{Construction of the Dual Solution $y$.} We start with some standard terminology of the dual-fitting argument. 
Consider a vertex $v\in V$, and let $u_{1},u_{2},...,u_{2n}$ be an order of vertices with $\dist_{\calM}(v,u_{j})$ increasing (in particular $u_{1}=v$). By \emph{growing a ball with radius $r$ around $v$}, we mean for each $u_{j}$ with $\dist_{\calM}(v,u_{j})\leq r$, adding a value $\min\{\dist_{\calM}(v,u_{j+1}),r\}-\dist_{\calM}(v,u_{j})$ to the dual variable $y(\{u_{1},...,u_{j}\})$.
The following observation is straightforward by the triangle inequality.
\begin{observation}
\label{ob:DualSolutionBallGrowing}
Let $X\subseteq V$ be a set of vertices with pairwise distances at least $2r$ in $\calM$. Then the dual solution obtained by growing a ball of radius $r$ around each vertex in $X$ satisfies all the dual constraints.
\end{observation}

The dual solution we will construct is exactly by growing a ball of radius $r$ around each vertex in a subset $X\subseteq V$. We set the radius 
\[
r = 2^{i-1},
\]
and define the subset $X\subseteq V$ by the following procedure.

The procedure will define subsets $\hat{X}^{(t)}$ and $X^{(t)}$ for each arrival $t$.
We will take the final $X$ to be $X^{(n)}$ (corresponding to the last arrival).
The subsets will be defined so as to
satisfy the following invariants:
\begin{enumerate}
\item\label{Invariant1} $\hat{X}^{(t)}$ is contained in the union of $i$-active clusters in $\scrC^{(t)}_{i+1}$. Recall that $i$-active clusters are the clusters $C$ with $\level(C)\geq i$.
\item\label{Invariant2} Vertices in $\hat{X}^{(t)}$ have pairwise distances at least $2r$ in $\calM$.
\item\label{Invariant3} $X^{(t)}\subseteq \hat{X}^{(t)}$, and for each $i$-active cluster $C\in \scrC^{(t)}_{i+1}$, $\hat{X}^{(t)}\cap C$ has at least one vertex outside $X^{(t)}$.
\item\label{Invariant4} $|X^{(t)}| = \sum_{t'\leq t}|\hat{F}^{(t')}_{i}\setminus \hat{F}^{(t')}_{\inh,i}|$.
\end{enumerate}
In particular, if some arrival $t$ has maximum level $L^{(t)}$ less than $i$, then $\hat{X}^{(t)}$ and $X^{(t)}$ are simply empty, which clearly satisfy all the invariants. In what follows, we first describe the construction of $\hat{X}^{(t)}$ and $X^{(t)}$, and then argue why taking $X:=X^{(n)}$ gives a feasible dual solution.

We now construct $\hat{X}^{(t)}$ and $X^{(t)}$ assuming that $\hat{X}^{(t-1)}$ and $X^{(t-1)}$ (satisfying the invariants) are given. Recall the following from \Cref{sect:OnlineSolution}: $\scrC^{(t)}_{i+1}$ is exactly the clustering obtained by contracting $\hat{F}^{(t)}_{i}$ over $\scrC^{(t)}_{i}$. The virtual edges $\hat{F}^{(t)}_{\inh,i}\subseteq \hat{F}^{(t)}_{i}$ are inherited. $\scrC^{(t)}_{\inh,i}$ is an intermediate clustering obtained by contracting $\hat{F}^{(t)}_{\inh,i}$ over $\scrC^{(t)}_{i}$. 

\medskip

\noindent\underline{Useful Properties on Clusterings.} The argument below will heavily exploit some useful properties of clusterings $\scrC^{(t-1)}_{i}, \scrC^{(t-1)}_{i+1}, \scrC^{(t)}_{i},\scrC^{(t)}_{\inh,i}$ and $\scrC^{(t)}_{i+1}$. First, directly from the clustering procedure, we have
\[
\scrC^{(t-1)}_{i}\preceq\scrC^{(t-1)}_{i+1},\qquad\text{and}\qquad \scrC^{(t)}_{i}\preceq\scrC^{(t)}_{\inh,i}\preceq\scrC^{(t)}_{i+1}.
\]
Second, recall that \Cref{ob:ClusteringRelationSameLevel} and \Cref{lemma:ClusteringRelationInherit} provide
\[
\scrC^{(t-1)}_{i}\preceq \scrC^{(t)}_{i},\qquad\text{and}\qquad\scrC^{(t-1)}_{i+1}\preceq \scrC^{(t)}_{\inh,i}.
\]
We sometimes consider the union of $i$-active clusters in these clusterings. Let $\ActUn(\scrC')$ denote the union of $i$-active clusters in a clustering $\scrC'$. We can easily observe that
\[
\ActUn(\scrC^{(t-1)}_{i}) = \ActUn(\scrC^{(t-1)}_{i+1})\subseteq \ActUn(\scrC^{(t)}_{i}) = \ActUn(\scrC^{(t)}_{\inh,i}) = \ActUn(\scrC^{(t)}_{i+1})
\]
by the clustering procedure and the above properties.

\medskip

\noindent\underline{Construction of $\hat{X}^{(t)}$.} Let $(u^{(t)},v^{(t)})$ be the new demand pair of arrival $t$. Initially, set $\hat{X}^{(t)} = \hat{X}^{(t-1)}$. If $\level(u^{(t)}),\level(v^{(t)})\geq i$, let $C_{u},C_{v}\in \scrC^{(t)}_{\inh,i}$ be the clusters containing $u^{(t)}$ and $v^{(t)}$ respectively (possibly $C_{u}$ and $C_{v}$ are the same cluster). 
\begin{itemize}
\item If $C_{u}$ and $C_{v}$ are the same cluster, add $u^{(t)}$ into $\hat{X}^{(t)}$ if $u^{(t)}$ and $v^{(t)}$ are the only vertices with level at least $i$ in $C_{u}$.
\item If $C_{u}$ and $C_{v}$ are different clusters, add $u^{(t)}$ (resp. $v^{(t)}$) into $\hat{X}^{(t)}$ if $u^{(t)}$ (resp. $v^{(t)}$) are the only vertices with level at least $i$ in $C_{u}$ (resp. $C_{v}$).
\end{itemize}

\begin{lemma}
\label{claim:Xhat}
We have the following.
\begin{enumerate}[label=(\alph*)]
\item\label{item:Xhat1} $\hat{X}^{(t)}$ is contained in the union of $i$-active clusters in $\scrC^{(t)}_{\inh,i}$, i.e., $\hat{X}^{(t)}\subseteq \ActUn(\scrC^{(t)}_{\inh,i})$.
\item\label{item:Xhat2} Vertices in $\hat{X}^{(t)}$ have pairwise distances at least $2r$ in $\calM$.
\item\label{item:Xhat3} For each $i$-active cluster $C\in\scrC^{(t)}_{\inh,i}$, $\hat{X}^{(t)}\cap C$ has at least one vertex outside $X^{(t-1)}$.
\end{enumerate}
\end{lemma}
\begin{proof}
\Cref{item:Xhat1} is straightforward from the construction of $\hat{X}^{(t)}$: the new vertices of $\hat{X}^{(t)}$ (i.e., one or more of $u^{(t)}$ and $v^{(t)}$) always come from the $i$-active clusters $C_{u},C_{v}\in \scrC^{(t)}_{\inh,i}$; the old vertices of $\hat{X}^{(t)}$ satisfy $\hat{X}^{(t-1)}\subseteq \ActUn(\scrC^{(t-1)}_{i+1})$ (by Invariant \ref{Invariant1}), and thus $\hat{X}^{(t-1)}\subseteq \ActUn(\scrC^{(t)}_{\inh,i})$.

Next, we show \Cref{item:Xhat2}. Suppose $C_{u}$ and $C_{v}$ are different clusters (an analogous argument works for the case $C_{u}=C_{v}$). It suffices to show that if we add $u^{(t)}$ into $\hat{X}^{(t)}$, this new $u^{(t)}$ is far from any vertices in $\hat{X}^{(t-1)}$. 
Recall that we add $u^{(t)}$ only if $u^{(t)}$ is the only vertex with level at least $i$ in $C_{u}$. This means the $i$-active cluster $C'_{u}\in \scrC^{(t)}_{i}$ containing $u^{(t)}$ also has $u^{(t)}$ as the only vertex with level at least $i$ (since $C'_{u}\subseteq C_{u}$), and thus, $C'_{u}$ is disjoint from any $i$-active cluster in $\scrC^{(t-1)}_{i}$ (since $\scrC^{(t-1)}_{i}\preceq \scrC^{(t)}_{i}$). Therefore, $C'_{u}$ is disjoint from $\ActUn(\scrC^{(t-1)}_{i}) = \ActUn(\scrC^{(t-1)}_{i+1})\supseteq \hat{X}^{(t-1)}$.

Now consider any vertex $x\in \hat{X}^{(t-1)}$. It must belong to an $i$-active $\scrC^{(t)}_{i}$-cluster $C'_{x}$ (since $\hat{X}^{(t-1)}\subseteq \ActUn(\scrC^{(t-1)}_{i+1})\subseteq \ActUn(\scrC^{(t)}_{i})$), and this $C'_{x}$ cannot be the same as $C'_{u}$ (since $C'_{u}$ is disjoint from $\hat{X}^{(t-1)}$). Hence, we can conclude
\[
\dist_{\calM}(x,u^{(t)})\geq \dist_{\calM/\scrC^{(t)}_{i}}(C'_{x},C'_{u})\geq 2^{i}
\]
as desired, where the second inequality is by \Cref{lemma:ClustersGap}.

Regarding \Cref{item:Xhat3}, if an $i$-active cluster $C\in\scrC^{(t)}_{\inh,i}$ contains an $i$-active cluster in $\scrC^{(t-1)}_{i+1}$, then by Invariant \ref{Invariant3} of $\hat{X}^{(t-1)}$, $\hat{X}^{(t-1)}\cap C$ will have at least one vertex outside $X^{(t-1)}$. Thus, we only need to pay attention to those $i$-active clusters $C\in \scrC^{(t)}_{\inh,i}$ containing no $i$-active cluster in $\scrC^{(t-1)}_{i+1}$. Because $\scrC^{(t-1)}_{i+1}\preceq\scrC^{(t)}_{\inh,i}$ (\Cref{lemma:ClusteringRelationInherit}), such a cluster $C$ must have no vertex with level at least $i$ from $V^{(t-1)}$. In other words, such a cluster $C$ exists only when $\level(u^{(t)}) = \level(v^{(t)})\geq i$ (since $C$ is $i$-active) and it must contain either $u^{(t)}$ or $v^{(t)}$ (or both). Then by our construction of $\hat{X}^{(t)}$, we indeed add a new vertex from $C$ into $\hat{X}^{(t)}$ as desired.
\end{proof}

\medskip

\noindent\underline{Construction of $X^{(t)}$.} Initially, set $X^{(t)} = X^{(t-1)}$.

\begin{observation}
Each $i$-active cluster in $\scrC^{(t)}_{\inh,i}$ is contained by an $i$-active cluster in $\scrC^{(t)}_{i+1}$. Each $i$-active cluster in $\scrC^{(t)}_{i+1}$ only contains $i$-active clusters in $\scrC^{(t)}_{\inh,i}$.
\end{observation}
\begin{proof}
This is simply because, from $\scrC^{(t)}_{\inh,i}$ to $\scrC^{(t)}_{i+1}$, the clustering procedure only merges $i$-active clusters.
\end{proof}

For each $i$-active cluster $C\in \scrC^{(t)}_{i+1}$, if it is made up of $k$ many $i$-active clusters in $\scrC^{(t)}_{\inh,i}$, then by \Cref{item:Xhat3} in \Cref{claim:Xhat}, $\hat{X}^{(t)}\cap C$ has at least $k$ vertices outside $X^{(t-1)}$, and we add arbitrary $k-1$ of them into $X^{(t)}$.

\medskip

\noindent\underline{Proving Invariants of $\hat{X}^{(t)}$ and $X^{(t)}$.} Invariant \ref{Invariant1} is by \Cref{item:Xhat1} in \Cref{claim:Xhat} and the fact that the union of $i$-active clusters in $\scrC^{(t)}_{\inh,i}$ is the same as that in $\scrC^{(t)}_{i+1}$. Invariant \ref{Invariant2} has also been proven in \Cref{claim:Xhat}. Invariant \ref{Invariant3} is by the construction of $X^{(t)}$. 

Finally, to see Invariant \ref{Invariant4}, it suffices to show that $|X^{(t)}|-|X^{(t-1)}| = |\hat{F}^{(t)}_{i}| - |\hat{F}^{(t)}_{\inh,i}|$. Furthermore, we can observe that $|\hat{F}^{(t)}_{i}| = |\scrC^{(t)}_{i}| - |\scrC^{(t)}_{i+1}|$, since $\scrC^{(t)}_{i+1}$ is obtained by contracting $\hat{F}^{(t)}_{i}$ over $\scrC^{(t)}_{i}$, and similarly, we have $|\hat{F}^{(t)}_{\inh,i}| = |\scrC^{(t)}_{i}| - |\scrC^{(t)}_{\inh,i}|$. By the construction of $X^{(t)}$, the number of new vertices added to $X^{(t)}$, i.e., $|X^{(t)}|-|X^{(t-1)}|$, is exactly the number of $i$-active clusters in $\scrC^{(t)}_{\inh,i}$ minus the number of $i$-active clusters in $\scrC^{(t)}_{i+1}$, which is the same as $|\scrC^{(t)}_{\inh,i}| - |\scrC^{(t)}_{i+1}|$ since the $i$-inactive clusters in $\scrC^{(t)}_{\inh,i}$ are the same as those in $\scrC^{(t)}_{i+1}$. Combining all the above, we get $|X^{(t)}|-|X^{(t-1)}| = |\hat{F}^{(t)}_{i}| - |\hat{F}^{(t)}_{\inh,i}|$.

\paragraph{Feasibility of the Dual Solution $y$.} Recall that we have two steps to show the feasibility of the dual solution $y$. First, it is required that $y$ satisfies all the dual constraints, which is indeed the case by \Cref{ob:DualSolutionBallGrowing} and the way we construct $y$. Second, we also need to show that $y$ is zero outside the actual domain ${\cal S}$ (recall that ${\cal S}$ collects all cuts separating at least one cluster in $\scrC^{(n)}_{\max}$).

Consider a vertex $v\in X$. When we grow a ball with radius $r$ around $v$, the cuts $S$ for which we may add a positive value to $y(S)$ must satisfy $S\subseteq \Ball(v,r) = \{u\in V^{(n)}\mid \dist_{\calM}(v,u)\leq r\}$. However, by Invariants \ref{Invariant1}, \ref{Invariant2} and \ref{Invariant3}, in the cluster $C_{v}\in \scrC^{(n)}_{i+1}$ containing $v$, there must be another vertex $v'\in C_{v}\cap (\hat{X}^{(n)}\setminus X)$ such that $\dist_{\calM}(v,v')\geq 2r$. Hence $v'$ is outside all such cuts $S$, which means that all such $S$ will separate $C_{v}$. Finally, since $\scrC^{(n)}_{i+1}\preceq \scrC^{(n)}_{\max}$, all such $S$ will separate the $\scrC^{(n)}_{\max}$-cluster containing $C_{v}$, as desired.

\paragraph{Completing the Proof.} By Invariant \ref{Invariant4}, the value of the dual solution $y$ is 
\[
D = \sum_{S}y(S) = |X|r = \sum_{t}|\hat{F}^{(t)}_{i}\setminus \hat{F}^{(t)}_{\inh,i}|\cdot 2^{i-1}.
\]
Combining it with $D\leq \OPT_{\LP}\leq O(\OPT)$, we get the original lemma.

\end{proof}

\section{Conclusion} \label{sec:future}

In this work we initiate the study of low-recourse algorithms for online Steiner forest.
We gave a constant-competitive algorithm with $O(\log n)$ amortized recourse.
This prompts several natural follow-up questions:

\begin{question}
    Is there an algorithm with $O(1)$ recourse?
\end{question}
\begin{question} \label{q2}
    Can the recourse be de-amortized?
\end{question}
\begin{question} \label{q4}
    Can one handle the fully-dynamic setting, where terminal pairs both arrive and depart?
    (What about the deletion-only model?)
\end{question}

For completeness, we also remark that for the Steiner \emph{tree} problem,
an algorithm with worst-case (un-amortized) constant recourse is not known for the fully dynamic setting.

\begin{remark}    
Obtaining $O(\log n)$ recourse in the \emph{deletion-only setting} (of online Steiner forest) is straightforward, at least if we assume that the optimal cost is bounded polynomially in $n$: one can recompute the solution whenever the optimal cost decreases by a factor 2.\footnote{%
We can further remove the assumption by \emph{bucketing} the demand pairs. That is, initially, we put a pair $(u,v)$ with $\dist_{\calM}(u,v)\in[n^{i},n^{i+1})$ into the $i$-th bucket, and then for each bucket $i$, create its own initial solution $F_{i}$ (let the entire initial solution be the union of all $F_{i}$). Note that $|F_{i}|$ should be proportional to the bucket size, since the original graph is complete with edge weights forming a metric.
With this initial solution, whenever the optimal cost is halved, we only need to recompute the $F_{i}$ of the top $O(1)$ buckets. Each $F_{i}$ is recomputed $O(\log n)$ times, leading to $O(\log n)$ amortized recourse.
}
\end{remark}

\section*{Acknowledgements}

The authors want to thank Roie Levin and Anupam Gupta for helpful discussions.

\appendix

\section{The Online (Timed) Gluttonous Algorithm}
\label{sect:OnlineGlut}

In this section, we give the formal description of an online simulation of the timed gluttonous algorithm in \cite{DBLP:conf/stoc/Gupta015}, which is a classic offline Steiner forest algorithm. We can show that it is $O(\log n)$-competitive.

\begin{algorithm}[H]
\caption{Online (Timed) Gluttonous}
\label{algo:OnlineGluttonous}
\begin{algorithmic}[1]
\State Define $\level(\cdot)$ for terminals and clusters as in \Cref{sect:Clustering}.
\State Initialize $\scrC$ to be an empty clustering of the initially empty terminal set.
\While{a demand pair $(u_{t},v_{t})$ arrives}
\State Add two clusters $\{u_{t}\},\{v_{t}\}$ into $\scrC$
\For{$i=0$ to the max level}
\While{exist $i$-active clusters  $C_{1},C_{2}\in \scrC$ s.t.~$\dist_{\calM/\scrC}(C_{1},C_{2})<2^{i+1}$}
\State In $\scrC$, replace $C_{1}$ and $C_{2}$ with $C_{1}\cup C_{2}$
\EndWhile
\EndFor
\EndWhile
\end{algorithmic}
\end{algorithm}

We can use the same idea as in the proof of \Cref{lemma:OnlineCost}
to prove an analogue of \Cref{lemma:OnlineCost}: for each level $i$, the number of level-$i$ merges times $2^{i+1}$ is at most $O(\OPT)$. Therefore, this algorithm is $O(\log n)$-competitive since the total cost of merges from outside of the top $O(\log n)$ levels is negligible (since each level has at most $O(n)$ merges, and the cost per merge decreases exponentially as $i$ goes down). 

\bibliographystyle{alpha}
\bibliography{ref}

\end{document}